\documentclass[a4paper,11pt]{article}
\input{Preamble.tex}

\begin{document}

\pagenumbering{Roman}

\hypersetup{pageanchor=false}
\title{Connectivity augmentation is fixed-parameter tractable
\thanks{Supported by the VILLUM Foundation, Grant Number 54451, Basic Algorithms Research Copenhagen (BARC). T.K. supported by the European Union under Marie Skłodowska-Curie Actions (MSCA), project no. 101206430.}
}

\author{Tuukka Korhonen\thanks{University of Copenhagen. Emails: \texttt{tuko@di.ku.dk}, \texttt{mthorup@di.ku.dk}}
\and
Mikkel Thorup\addtocounter{footnote}{-1}\footnotemark{}
}


\maketitle

\thispagestyle{empty}

\begin{abstract}
In the vertex connectivity augmentation problem, we are given an undirected $n$-vertex graph $G$, a set of links $L \subseteq \binom{V(G)}{2} \setminus E(G)$, and integers $\lambda$ and $k$.
The task is to insert at most $k$ links from $L$ to $G$ to make $G$ $\lambda$-vertex-connected.
We show that the problem is fixed-parameter tractable (FPT) when parameterized by $\lambda$ and $k$, by giving an algorithm with running time $2^{\OO(k \log (k + \lambda))} n^{\OO(1)}$.
This improves upon a recent result of Carmesin and Ramanujan~[SODA'26], who showed that the problem is FPT parameterized by $k$ but only when~$\lambda \le 4$.

We also consider the analogous edge connectivity augmentation problem, where the goal is to make $G$ $\lambda$-edge-connected.
We show that the problem is FPT when parameterized by $k$ only, by giving an algorithm with running time $2^{\OO(k \log k)} n^{\OO(1)}$.
Previously, such results were known only under additional assumptions on the edge connectivity of $G$.
\end{abstract}

\begin{textblock}{20}(-0.5, 7.9)
\includegraphics[width=100px]{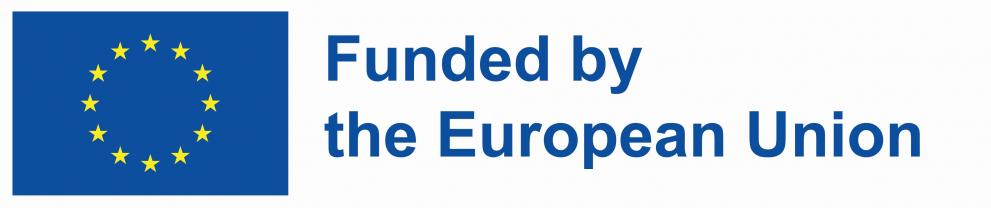}%
\end{textblock}

\thispagestyle{empty}

\newpage


\setcounter{page}{1}



\pagenumbering{arabic}

\hypersetup{pageanchor=true}

\clearpage
\setcounter{page}{1}

\section{Introduction}
In connectivity augmentation problems the task is to add a smallest-cost set of edges to a graph so that it achieves a desired connectivity value $\lambda$.
Connectivity augmentation problems were introduced by Eswaran and Tarjan~\cite{DBLP:journals/siamcomp/EswaranT76} and Plesn{\'{\i}}k~\cite{plesenik1976minimum} in the 1970s, and have a rich literature of polynomial-time exact algorithms~\cite{DBLP:journals/siamcomp/RosenthalG77,DBLP:journals/jcss/WatanabeN87,DBLP:journals/siamdm/Frank92,DBLP:journals/siamdm/Vegh11}, approximation algorithms~\cite{DBLP:journals/siamcomp/FredericksonJ81,DBLP:journals/dam/Nagamochi03,DBLP:journals/siamcomp/ByrkaGA23,hommelsheim2026better}, and parameterized algorithms~\cite{DBLP:journals/dam/Nagamochi03,DBLP:journals/jct/JacksonJ05a,DBLP:journals/networks/GuoU10,DBLP:journals/talg/MarxV15,DBLP:conf/icalp/BasavarajuFGMRS14,DBLP:conf/esa/Nutov24,aug4vc}.
In this paper, we focus on parameterized algorithms for the two most basic variants of connectivity augmentation problems that are known to be NP-hard.

A graph $G$ is $\lambda$-vertex-connected (resp. $\lambda$-edge-connected) if it is not possible to remove $< \lambda$ vertices (resp. edges) to split $G$ into more than one connected component.
In the \emph{vertex connectivity augmentation} problem, we are given a graph $G$, a set of \emph{links} $L \subseteq \binom{V(G)}{2} \setminus E(G)$, connectivity target $\lambda$, and budget $k$, and the task is to add at most $k$ links from $L$ to $E(G)$ to make $G$ $\lambda$-vertex-connected.
The \emph{edge connectivity augmentation} problem is defined similarly, but the goal is to make $G$ $\lambda$-edge-connected.

When the set $L$ includes all possible extra edges, i.e., $L = \binom{V(G)}{2} \setminus E(G)$, the connectivity augmentation problems become amenable to surprisingly powerful structural theory.
The edge-version was shown to be polynomial-time solvable in this setting by Watanabe and Nakamura~\cite{DBLP:journals/jcss/WatanabeN87} (see also~\cite{DBLP:journals/siamdm/Frank92}).
The vertex-version was shown to be fixed-parameter tractable (FPT) parameterized by $\lambda$ by Jackson and Jord{\'{a}}n~\cite{DBLP:journals/jct/JacksonJ05a}, and polynomial-time solvable when $G$ is $(\lambda-1)$-vertex-connected by V{\'{e}}gh~\cite{DBLP:journals/siamdm/Vegh11}.
It remains an open problem whether it is polynomial-time solvable or NP-hard in general.

When the set $L$ is allowed to be arbitrary the problems become harder.
It is easy to observe that both vertex and edge connectivity augmentation are NP-hard already when $\lambda = 2$, as Hamiltonian cycle reduces to such case with $E(G) = \emptyset$.
This hardness also holds when $G$ is a tree, and extends to higher values of $\lambda$~\cite{DBLP:journals/siamcomp/FredericksonJ81}.

Parameterized algorithms for the case of arbitrary $L$ were first considered by Guo and Uhlmann~\cite{DBLP:journals/networks/GuoU10}, who observed that the work of Nagamochi~\cite{DBLP:journals/dam/Nagamochi03} implies that edge connectivity augmentation is FPT parameterized by $k$ when $\lambda = 2$ and $G$ is connected.
Guo and Uhlmann themselves showed that the problem and its vertex-version admit polynomial kernelization, establishing that the vertex-version is also FPT by $k$, under the same assumptions that $\lambda = 2$ and $G$ is connected.
Marx and V{\'{e}}gh~\cite{DBLP:journals/talg/MarxV15} extended these results in various ways: They showed that edge connectivity augmentation is FPT parameterized by $k$ and has a polynomial kernel whenever $G$ is $(\lambda-1)$-connected, and edge connectivity augmentation is FPT parameterized by $k$ for $\lambda = 2$ even if $G$ is not connected.
Their algorithms work also for the weighted versions of these problems.
Basavaraju, Fomin, Golovach, Misra, Ramanujan, Saurabh~\cite{DBLP:conf/icalp/BasavarajuFGMRS14} gave a faster FPT algorithm for the case when $G$ is $(\lambda-1)$-connected.

Later, Nutov~\cite{DBLP:conf/esa/Nutov24} showed that vertex connectivity augmentation is FPT parameterized by $k$ when $\lambda = 3$ and $G$ is $2$-vertex-connected.
Carmesin and Ramanujan~\cite{aug4vc} improved this result by showing that vertex connectivity augmentation is FPT parameterized by $k$ for $\lambda \le 4$, regardless of the connectivity of the input graph.
By a reduction, results for vertex connectivity augmentation imply also the same results for edge connectivity augmentation. Thus, for both vertex and edge connectivity augmentation, the problem was only known to be FPT for $\lambda\leq 4$.

In this paper, we present a significantly more general algorithm showing that vertex connectivity augmentation is FPT when parameterized by $k$ and $\lambda$.

\begin{restatable}{theorem}{thmmain}
\label{thm:main}
There is an algorithm for the vertex connectivity augmentation problem that runs in time $2^{\OO(k \log (k+\lambda))} |G|^{\OO(1)}$.
\end{restatable}

Note that when $\lambda$ is small enough as a function of $n$, in particular, when $\lambda \le n^{o(1)}$, then the algorithm of \Cref{thm:main} is in fact FPT parameterized by $k$ only.\footnote{For example, if $\lambda \le n^{1/\log \log n}$, the algorithm would run in time $n^{\OO(k \log k/\log \log n)} \le 2^{2^{\OO(k \log k)}} + n^{\OO(1)}$.}

While \Cref{thm:main} directly implies the same result for edge connectivity augmentation, in the case of edge connectivity we can improve the algorithm and remove the dependency on $\lambda$ from it entirely, showing that it is FPT parameterized by $k$ only.

\begin{restatable}{theorem}{thmedgeconn}
\label{thm:edgeconn}
There is an algorithm for the edge connectivity augmentation problem that runs in time $2^{\OO(k \log k)} |G|^{\OO(1)}$.
\end{restatable}

The algorithms of \Cref{thm:main,thm:edgeconn} are the first to show that vertex and edge connectivity augmentation are FPT parameterized by $k$ and $\lambda$ without any assumptions on the connectivity of $G$.

\paragraph{Outline of our algorithms.}
The algorithms of \Cref{thm:main,thm:edgeconn} are based on finding a ``relevant'' subset of links $L' \subseteq L$, so that (1) $|L'|$ is bounded, and (2) if a solution exists, then a solution that contains a link from $L'$ exists.
After finding such $L'$, it is easy to obtain an FPT algorithm by branching.
Of course, this is a classic approach for designing FPT algorithms, which has also been used in the context of connectivity augmentation before~\cite{DBLP:journals/dam/Nagamochi03,DBLP:journals/networks/GuoU10,aug4vc}, so the novelty of our work is in finding the relevant subset $L'$.

We find $L'$ by repeatedly removing ``irrelevant'' links from $L$.
The first case, in both the vertex and edge versions, is when there is a vertex $v$ with a large number of links incident to it.
We show that then, either (1) one of the links is dominated by the others, in that it can be always replaced by a better link, or (2) the endpoints of the links can be separated too well from each other, showing that there is in fact no solution with $\le k$ links.
The other case is when there are no high-degree vertices with respect to the links, but still many links.
In this case, we identify a ``leaf-like'' region, which consists of vertices that are well-connected to each other but to which we must attach links 
to better connect it to a rest of the graph.
With a bit more complicated arguments, the leaf-like region can then play the role of the single vertex $v$ in the previous argument, and we can identify irrelevant
links incident to it until there is only bounded number of them.

Our algorithms for vertex connectivity augmentation and edge connectivity augmentation are very similar to each other.
The main difference is, that for edge connectivity augmentation we can get down to $\OO(k)$ relevant links, while for vertex connectivity augmentation we can only get down to $\OO(k \lambda)$ relevant links.
The reason for this is that for vertex connectivity augmentation, the vertices in the separators play a special role, incurring a factor of $\OO(\lambda)$ from the separator size.


\paragraph{Comparison to~\cite{aug4vc}.}
On a high-level, our approach is similar to that of Carmesin and Ramanujan~\cite{aug4vc} and inspired by it.
The main difference is that in the arguments for arguing irrelevance of links and detecting the ``leaf-like'' region, Carmesin and Ramanujan use a special graph decomposition that describes the structure of separators of size $< 4$ very precisely.
This structure is not known to extend to larger sizes of separators.
We instead use more approximate arguments, that do not require a precise understanding of all separators of size $< \lambda$, but instead employ Ramsey/pigeonhole-style reasoning of finding a well-structured subset in a large unstructured set.
This argumentation turns out to be simpler and lead to more general results.

\paragraph{Organization of the paper.}
We start by presenting definitions and preliminary results in \Cref{sec:prelims}.
Then, in \Cref{sec:vconnaug} we present our algorithm for vertex connectivity augmentation, i.e., prove \Cref{thm:main}.
In \Cref{sec:edgeconnaug} we present our algorithm for edge connectivity augmentation, i.e., prove \Cref{thm:edgeconn}.
We conclude with open problems in \Cref{sec:conclusions}.

\section{Preliminaries}
\label{sec:prelims}
For a positive integer $n$, we denote the set of integers $1, \ldots, n$ by $[n]$.
For a set $X$, we denote by $\binom{X}{2}$ the collection of all unordered pairs of elements of $X$.
In particular, $|\binom{X}{2}| = \binom{|X|}{2}$.

In this paper, we consider undirected simple graphs.
The vertex set of a graph $G$ is denoted by $V(G)$ and the edge set by $E(G) \subseteq \binom{V(G)}{2}$.
The open neighborhood of a vertex $v \in V(G)$ is denoted by $N(v) = \{u \mid vu \in E(G)\}$, and the closed neighborhood is $N[v] = N(v) \cup \{v\}$.
For a graph $G$ and a set $X \subseteq \binom{V(G)}{2}$, we denote by $G \cup X$ the graph with the vertex set $V(G)$ and edge set $E(G) \cup X$.
Analogously, $G \setminus X$ denotes the graph with the vertex set $V(G)$ and edge set $E(G) \setminus X$.
An independent set of $G$ is a set of vertices $I \subseteq V(G)$ so that there are no edges between vertices in $I$.
A matching is a set of edges $M \subseteq E(G)$ that do not share a vertex.

\subsection{Vertex separations}
A \emph{separation} of a graph $G$ is a pair of vertex sets $(A,B)$, so that $A \cup B = V(G)$ and there are no edges between $A \setminus B$ and $B \setminus A$.
The \emph{separator} of a separation $(A,B)$ is the set $A \cap B$, and the \emph{order} of it is $|A \cap B|$.
A separation is \emph{proper} if both $A \setminus B$ and $B \setminus A$ are non-empty.
For two distinct vertices $a, b \in V(G)$, an \emph{$(a,b)$-separation} is a separation $(A,B)$ with $a \in A \setminus B$ and $b \in B \setminus A$.
Note that for an $(a,b)$-separation $(A,B)$, it must hold that $N[a] \subseteq A$ and $N[b] \subseteq B$.

Two distinct vertices $a,b \in V(G)$ are $\lambda$-vertex-connected if there is no $(a,b)$-separation of order $<\lambda$.
In particular, if $ab \in E(G)$, then $a$ and $b$ are $\lambda$-vertex-connected for all $\lambda$.
A graph $G$ is $\lambda$-vertex-connected if all pairs of its vertices are $\lambda$-vertex-connected, or equivalently, if there is no proper separation of order $<\lambda$.
The vertex connectivity of $G$ is the maximum $\lambda$ so that $G$ is $\lambda$-vertex-connected, or $\infty$ if $G$ is $\lambda$-vertex-connected for all $\lambda$, in particular, if $G$ is a clique.

We note that often in the literature, $\lambda$-vertex-connectivity is defined so that it also requires $|V(G)| \ge k+1$.
We use our definition as it is more natural for our proofs, but note that for \Cref{thm:main,thm:edgeconn} it makes no difference which definition we choose.

We need the following observation on vertex connectivity.


\begin{observation}
\label{obs:vertconn}
If $G$ is a graph and $e \in \binom{V(G)}{2}$, then the vertex connectivity of $G \cup \{e\}$ is at most one more than the vertex connectivity of $G$, unless $G \cup \{e\}$ is a clique.
\end{observation}

The following is a well-known observation on separations of graphs.

\begin{observation}
\label{obs:genseps}
If $(A,B)$ and $(C,D)$ are separations of $G$, then also $(A \cap C, B \cup D)$ is a separation of $G$.
\end{observation}
\begin{proof}
Let $X = (A \cap C) \setminus (B \cup D) = (A \setminus B) \cap (C \setminus D)$ and $Y = (B \cup D) \setminus (A \cap C) = (B \setminus A) \cup (D \setminus C)$.
It follows that if there is an edge between $X$ and $Y$, then it is either between $B \setminus A$ and $A \setminus B$, or between $D \setminus C$ and $C \setminus D$.
\end{proof}

Note that \Cref{obs:genseps} can be used to generate up to four different separations by permuting the orders of the pairs $(A,B)$ and $(C,D)$.

The next lemma is the well-known property of \emph{submodularity} of separations.

\begin{lemma}[submodularity of separations]
\label{lem:vertsubmod}
If $(A,B)$ and $(C,D)$ are separations of $G$, then 
\[|(A \cap C) \cap (B \cup D)| + |(A \cup C) \cap (B \cap D)| \le |A \cap B| + |C \cap D|.\]
\end{lemma}
\begin{proof}
Let $S_1 = (A \cap C) \cap (B \cup D)$ and $S_2 = (A \cup C) \cap (B \cap D)$, and observe that $S_1 \cup S_2 \subseteq (A \cap B) \cup (C \cap D)$.
If a vertex is in both $S_1$ and $S_2$, then it is in $(A \cap C) \cap (B \cup D) \cap (A \cup C) \cap (B \cap D) = A \cap B \cap C \cap D$, so it is in both $A \cap B$ and $C \cap D$.
This implies the conclusion.
\end{proof}

The next lemma states that there are unique ``left-most'' and ``right-most'' minimum $(a,b)$-separations.

\begin{lemma}
\label{lem:leftrightsep}
Let $G$ be a graph, $a,b \in V(G)$ distinct non-adjacent vertices, and $\lambda$ the minimum order of an $(a,b)$-separation of $G$.
There exists $(a,b)$-separations $(A_1, B_1)$ and $(A_2, B_2)$ of order $\lambda$ (possibly the same), so that for all $(a,b)$-separations $(A,B)$ of order $\lambda$ it holds that $A_1 \subseteq A \subseteq A_2$ and $B_2 \subseteq B \subseteq B_1$.
Moreover, such separations can be found in polynomial time.
\end{lemma}
\begin{proof}
Let $(A_1, B_1)$ be an arbitrary $(a,b)$-separation of order $\lambda$ that minimizes $|A_1|$.
Note that $(A_1, B_1)$ not only minimizes $|A_1|$, but also maximizes $|B_1|$ among such separations, as we know that $|B_1| = |V(G)| - |A_1| + \lambda$.
Consider also an arbitrary $(a,b)$-separation $(A,B)$ of order $\lambda$.

\begin{claim}
$A_1 \subseteq A$ and $B \subseteq B_1$.
\end{claim}
\begin{claimproof}
Consider the separations $(A_\ell, B_\ell) = (A \cap A_1, B \cup B_1)$ and $(A_r, B_r) = (A \cup A_1, B \cap B_1)$.
Both $(A_\ell, B_\ell)$ and $(A_r, B_r)$ are $(a,b)$-separations.
By submodularity we have that
\[|A_\ell \cap B_\ell| + |A_r \cap B_r| \le |A \cap B| + |A_1 \cap B_1|,\]
which by the fact that $(A,B)$ and $(A_1, B_1)$ are minimum order $(a,b)$-separations implies that $|A_\ell \cap B_\ell| = \lambda$ and $|A_r \cap B_r| = \lambda$.

Now, if $A_1 \not\subseteq A$, then $|A_\ell| < |A_1|$, contradicting the choice of $(A_1, B_1)$.
Similarly, if $B \not\subseteq B_1$, then $|B_\ell| > |B_1|$, also contradicting the choice of $(A_1, B_1)$.
\end{claimproof}

This proves the existence of such $(A_1, B_1)$.
By swapping the order of $a$ and $b$, we obtain the existence of $(A_2, B_2)$ with the same argument.

We recall that the separations $(A_1, B_1)$ and $(A_2, B_2)$ can be found in polynomial-time via the standard Ford-Fulkerson algorithm, in particular, the standard application of it for finding a minimum $(a,b)$-separation finds the separation $(A_1,B_1)$, and by swapping $a$ and $b$ we find $(A_2, B_2)$.
\end{proof}

\subsection{Edge cuts}
A \emph{cut} of a graph $G$ is a pair $(A,B)$ of disjoint sets with $A \cup B = V(G)$.
The edges $uv$ of $G$ with $u \in A$ and $v \in B$ are denoted by $E(A,B)$, and the \emph{order} of $(A,B)$ is $|E(A,B)|$.
A cut is \emph{proper} if both $A$ and $B$ are non-empty.
For vertices $a,b \in V(G)$, a cut $(A,B)$ is an $(a,b)$-cut if $a \in A$ and $b \in B$.
A set $X \subseteq E(G)$ is an \emph{$(a,b)$-cutset} if $a$ and $b$ are in different connected components of $G \setminus X$, in particular, if there exists an $(a,b)$-cut $(A,B)$ with $E(A,B) \subseteq X$.

Note that if $(A,B)$ and $(C,D)$ are cuts, then so are $(A \cap C, B \cup D)$ and $(A \cup C, B \cap D)$.
It is well-known that cuts are submodular.

\begin{lemma}[submodularity of cuts]
\label{lem:edgesubmod}
For two cuts $(A,B)$ and $(C,D)$ of $G$, it holds that
\[|E(A \cap C, B \cup D)| + |E(A \cup C, B \cap D)| \le |E(A,B)|+|E(C,D)|.\]
\end{lemma}
\begin{proof}
First, note that $E(A \cap C, B \cup D) \cup E(A \cup C, B \cap D) \subseteq E(A,B) \cup E(C,D)$.
Then, if an edge is in both $E(A \cap C, B \cup D)$ and $E(A \cup C, B \cap D)$, it must be between $A \cap C$ and $B \cap D$, so it is in both $E(A,B)$ and $E(C,D)$.
This implies the conclusion.
\end{proof}

The next lemma states that there are unique ``left-most'' and ``right-most'' minimum $(a,b)$-cuts.

\begin{lemma}
\label{lem:leftrightcut}
Let $G$ be a graph, $a,b \in V(G)$ distinct vertices, and $\lambda$ the minimum order of an $(a,b)$-cut of $G$.
There exists $(a,b)$-cuts $(A_1, B_1)$ and $(A_2, B_2)$ of order $\lambda$ (possibly the same), so that for all $(a,b)$-cuts $(A,B)$ of order $\lambda$ it holds that $A_1 \subseteq A \subseteq A_2$ and $B_2 \subseteq B \subseteq B_1$.
Moreover, such cuts can be found in polynomial time.
\end{lemma}
\begin{proof}
Let $(A_1, B_1)$ be an arbitrary $(a,b)$-cut of order $\lambda$ that minimizes $|A_1|$.
Consider an arbitrary $(a,b)$-cut $(A,B)$ of order $\lambda$.

\begin{claim}
$A_1 \subseteq A$.
\end{claim}
\begin{claimproof}
Suppose that $A_1 \not\subseteq A$ and consider the cuts $(A_\ell, B_\ell) = (A \cap A_1, B \cup B_1)$ and $(A_r, B_r) = (A \cup A_1, B \cap B_1)$.
Both of them are $(a,b)$-cuts.
By submodularity,
\[|E(A_\ell, B_\ell)| + |E(A_r, B_r)| \le |E(A,B)| + |E(A_1, B_1)|,\]
which by the fact that $(A,B)$ and $(A_1, B_1)$ are minimum order $(a,b)$-cuts implies that $|E(A_\ell, B_\ell)| = |E(A_r, B_r)| = \lambda$.

Now, $|A_\ell| < |A_1|$, contradicting the choice of $(A_1, B_1)$.
\end{claimproof}

Note that $A_1 \subseteq A$ implies $B \subseteq B_1$, so this gives the existence of such $(A_1, B_1)$.
By swapping the order of $a$ and $b$ in the above argument, we obtain the existence of $(A_2, B_2)$.

We recall that the cuts $(A_1, B_1)$ and $(A_2, B_2)$ can be found in polynomial-time via the standard Ford-Fulkerson algorithm, in particular, the standard application of it for finding a minimum $(a,b)$-cut finds the cut $(A_1,B_1)$, and by swapping $a$ and $b$ we find $(A_2, B_2)$.
\end{proof}

\subsection{Connectivity augmentation}
An \emph{instance} of the vertex connectivity augmentation problem is a $4$-tuple $\ins = (G, L, \lambda, k)$, where $G$ is a graph, $L \subseteq \binom{V(G)}{2} \setminus E(G)$ is a set of links, and $\lambda$ and $k$ are integers.
A solution to the instance $\ins$ is a subset $S \subseteq L$ of size $|S| \le k$ so that $G \cup S$ is $\lambda$-vertex-connected.
The instance $\ins$ is a \emph{yes-instance} if there exists a solution and \emph{no-instance} otherwise.

We define instances of the edge connectivity augmentation problem similarly, except replacing $\lambda$-vertex-connectedness by $\lambda$-edge-connectedness.

Let $\ins = (G, L, \lambda, k)$ be an instance of the vertex connectivity augmentation problem.
A link $uv \in L$ \emph{crosses} a separation $(A,B)$ of $G$ if $u \in A \setminus B$ and $v \in B \setminus A$.
Note that $S \subseteq L$ is a solution to $\ins$ if and only if for every proper separation $(A,B)$ of $G$ of order $< \lambda$, there is a link in $S$ that crosses $(A,B)$.

Similarly, if $\ins = (G, L, \lambda, k)$ is an instance of the edge connectivity augmentation problem, we define that a link $uv \in L$ \emph{crosses} a cut $(A,B)$ of $G$ if $u \in A$ and $v \in B$.
We observe that $S \subseteq L$ is a solution to $\ins$ if and only if for every proper cut $(A,B)$ of $G$ of order $|E(A,B)| < \lambda$, there are at least $\lambda-|E(A,B)|$ links in $S$ that cross $(A,B)$.


\section{Vertex connectivity augmentation}
\label{sec:vconnaug}
In this section we prove \Cref{thm:main}.
First we present an algorithm for finding a separation with certain properties, identifying a ``leaf-like'' region of the graph. Next we provide lemmas about finding ``irrelevant links'', and finally we put these ingredients together to give a branching algorithm.


\subsection{Separation with a well-connected left side}\label{sec:well-connected-left}
The goal of this subsection is to prove the following lemma for finding a well-connected ``leaf-like'' region in the input graph $G$.

\begin{lemma}
\label{lem:specialsepalg}
Let $G$ be a graph and $\lambda,k \ge 1$ integers so that $G$ is not $\lambda$-vertex-connected, but is $(\lambda-k)$-vertex-connected.
There exists a proper separation $(A,B)$ of $G$ of order $<\lambda$ so that vertices in $A \setminus B$ are pairwise $\lambda$-vertex-connected, and it can be found in $2^k |G|^{\OO(1)}$ time.
\end{lemma}



Note that the vertices in $A \setminus B$ are required to be pairwise $\lambda$-vertex-connected in the whole graph $G$, not in a subgraph of it.
We start by using a tool called \emph{unbreakable tree decomposition} to show the existence of such a separation.

A \emph{tree decomposition} of a graph $G$ is a pair $(T,\bag)$, where $T$ is a tree and $\bag \colon V(T) \to 2^{V(G)}$ a bag-function satisfying
\begin{enumerate}
\item $V(G) = \bigcup_{t \in V(T)} \bag(t)$,
\item $E(G) \subseteq \bigcup_{t \in V(T)} \binom{\bag(t)}{2}$, and
\item for each $v \in V(G)$, the set $\{t \in V(T) \mid v \in \bag(t)\}$ induces a connected subtree of $T$.
\end{enumerate}

For an edge $st \in E(T)$, we denote by $T_{\vec{st}}$ the subtree of $T$ consisting of the nodes closer to $s$ than $t$, and by $T_{\vec{ts}}$ the subtree of $T$ consisting of nodes closer to $t$ than $s$.
We denote by $\bag(T_{\vec{st}})$ the union of the bags of the nodes of $T_{\vec{st}}$.
The important property of tree decompositions is that for all $st \in E(T)$, the pair $(\bag(T_{\vec{st}}), \bag(T_{\vec{ts}}))$ is a separation of $G$.
The \emph{adhesion} at $st$ is the set $\adh(st) = \bag(s) \cap \bag(t) = \bag(T_{\vec{st}}) \cap \bag(T_{\vec{ts}})$, i.e., the separator of the separation.

For an integer $\lambda$, a set of vertices $X \subseteq V(G)$ is \emph{$\lambda$-unbreakable} if there is no separation $(A,B)$ of $G$ of order $< \lambda$ so that $|X \cap A| \ge \lambda$ and $|X \cap B| \ge \lambda$.
We use the following theorem, that was proven by Carmesin, Diestel, Hamann, and Hundertmark~\cite{DBLP:journals/siamdm/CarmesinDHH14} and Cygan, Komosa, Lokshtanov, Pilipczuk, Pilipczuk, Saurabh, and Wahlstr{\"{o}}m~\cite{DBLP:journals/talg/CyganKLPPSW21}, independently of each other but both based on Bellenbaum and Diestel~\cite{bellenbaum2002two}.

\begin{theorem}[\cite{DBLP:journals/siamdm/CarmesinDHH14,DBLP:journals/talg/CyganKLPPSW21}]
\label{thm:leantds}
For every graph $G$ and every integer $\lambda$, there exists a tree decomposition $(T,\bag)$ of $G$ so that $\bag(t)$ is $\lambda$-unbreakable for all $t \in V(T)$, and every adhesion of $(T,\bag)$ has size $< \lambda$.
\end{theorem}

Now we use \Cref{thm:leantds} to prove the existence of a desired separation.
The idea is that we can take a separation corresponding to a leaf of the tree decomposition.

\begin{lemma}
\label{lem:specialsep}
Let $G$ be a graph and $\lambda$ an integer.
If $G$ is not $\lambda$-vertex-connected, then $G$ has a proper separation $(A,B)$ of order $< \lambda$ so that vertices in $A \setminus B$ are pairwise $\lambda$-vertex-connected.
\end{lemma}
\begin{proof}
If $G$ has a vertex of degree $< \lambda$, then we can pick such a vertex $v$ with the smallest degree and let the separation be $(N[v], V(G) \setminus \{v\})$, unless $G$ is a clique in which case it is $\lambda$-vertex-connected.
Assume then that the minimum degree of $G$ is $\ge \lambda$.

By \Cref{thm:leantds}, let $(T,\bag)$ be a tree decomposition of $G$ so that $\bag(t)$ is $\lambda$-unbreakable for all $t \in V(T)$ and each adhesion of $(T,\bag)$ has size $< \lambda$.
Note that if such a tree decomposition has a leaf $\ell$ with neighbor $p$ with $\bag(\ell) \subseteq \bag(p)$, we can remove the leaf $\ell$.
Therefore we can assume that for each leaf $\ell$ with neighbor $p$, it holds that $\bag(\ell) \not\subseteq \bag(p)$.
This in particular implies that for all $st \in E(T)$, the separation $(\bag(T_{\vec{st}}), \bag(T_{\vec{ts}}))$ is proper.

If $(T,\bag)$ has only one bag, then $V(G)$ is $\lambda$-unbreakable, which implies that $G$ is $\lambda$-vertex-connected.
If $(T,\bag)$ has more than one bag, let $\ell \in V(T)$ be an arbitrary leaf with neighbor $p$.
We claim that the separation $(A,B) = (\bag(T_{\vec{\ell p}}), \bag(T_{\vec{p \ell}})) = (\bag(\ell), V(G) \setminus (\bag(\ell) \setminus \adh(\ell p)))$ satisfies the required properties.

As already observed, it is proper, and by the definition of $(T,\bag)$ has order $< \lambda$.
The vertices in $A \setminus B$ are the vertices that are only in the bag $\bag(\ell)$.
For any such vertex $v \in A \setminus B$, it must hold that $N[v] \subseteq \bag(\ell)$.
Suppose that two such vertices $u, v \in A \setminus B$ are not $\lambda$-vertex-connected, i.e., there exists a separation $(C,D)$ of $G$ of order $< \lambda$ with $u \in C \setminus D$ and $v \in D \setminus C$.
Now, $N[u] \subseteq C$ and $N[v] \subseteq D$, so $|C \cap \bag(\ell)| \ge |N[u]|$ and $|D \cap \bag(\ell)| \ge |N[v]|$.
Because $|N[u]| \ge \lambda$ and $|N[v]| \ge \lambda$, this contradicts the $\lambda$-unbreakability of $\bag(\ell)$.
\end{proof}

We then show that a separation whose existence is guaranteed by \Cref{lem:specialsep} can also be found efficiently enough.
We actually consider the special case when there are fixed vertices $a,b$ so that the separation is an $(a,b)$-separation, as that implies the general case.

\begin{lemma}
\label{lem:specialsepalgspecial}
Let $G$ be a graph, $a,b \in V(G)$ distinct non-adjacent vertices, and $\lambda,k \ge 1$ so that $a$ and $b$ are not $\lambda$-vertex-connected but are $(\lambda-k)$-vertex-connected.
Then, we can in $2^k |G|^{\OO(1)}$ find an $(a,b)$-separation $(A,B)$ of order $<\lambda$ so that vertices in $A \setminus B$ are pairwise $\lambda$-vertex-connected, if such a separation exists.
\end{lemma}
\begin{proof}
We design a recursive algorithm for the problem, where the parameter $k$ decreases on each recursion level.
The parameter $\lambda$ and the vertices $a$ and $b$ stay the same throughout the recursion, while the graph $G$ is changed by adding edges to it.

Let $\lambda' \ge \lambda-k$ be the minimum order of an $(a,b)$-separation.
We can assume that $\lambda' < \lambda$; otherwise we can return that no desired separation exists.
By \Cref{lem:leftrightsep}, let $(A_1, B_1)$ be the $(a,b)$-separation of order $\lambda'$ so that for all $(a,b)$-separations $(A,B)$ of order $\lambda'$ it holds that $A_1 \subseteq A$ and $B \subseteq B_1$.
The separation can be found in polynomial time.
We can also check in polynomial time if vertices in $A_1 \setminus B_1$ are pairwise $\lambda$-vertex-connected.
If yes, we have found the desired separation.
If no, let $u,v \in A_1 \setminus B_1$ so that $u$ and $v$ are not $\lambda$-vertex-connected.
Any desired separation $(A,B)$ must have either $u \in B$ or $v \in B$, so we recurse into two subproblems, one that forces $u$ into $B$ and one that forces $v$ into $B$.
In particular, the former corresponds to adding the edge $ub$ to $G$ and the latter to adding the edge $vb$ to $G$.

It is clear that this recursive procedure is correct.
To bound the running time, it suffices to show that we can decrease $k$ by one when calling the algorithm recursively.
For this, it suffices to show that if $v \in A_1 \setminus B_1$, and $G^v$ is the graph $G$ with the edge $vb$ added, then $a$ and $b$ are $(\lambda'+1)$-vertex-connected in $G^v$.

Suppose otherwise and let $(C,D)$ be an $(a,b)$-separation of $G^v$ of order $\lambda'$.
It is also an $(a,b)$-separation of $G$.
By the guarantee of \Cref{lem:leftrightsep}, we therefore have that $D \subseteq B_1$.
However, this is a contradiction, as $v \in D$ but $v \notin B_1$.
\end{proof}

Now we can put together the proof of \Cref{lem:specialsepalg}.

\begin{proof}[Proof of~\Cref{lem:specialsepalg}]
By \Cref{lem:specialsep} the desired separation $(A,B)$ exists.
It is a proper separation, and therefore an $a,b$-separation for some non-adjacent $a,b \in V(G)$.
By trying all pairs $a,b$ with the algorithm of \Cref{lem:specialsepalgspecial}, we can find such a separation.
\end{proof}

\subsection{Irrelevant links}\label{sec:irrel-links}
Let $\ins = (G, L, \lambda, k)$ be an instance of the vertex connectivity augmentation problem, and $L' \subseteq L$ a subset of the links.
A subset $L'' \subseteq L'$ is an \emph{$L'$-relevant} set if whenever there exists a solution to $\ins$ that contains a link from $L'$, there exists a solution to $\ins$ that contains a link from $L''$.
We say that a link $uv \in L'$ is \emph{$L'$-irrelevant} if $L' \setminus \{uv\}$ is $L'$-relevant.

Our algorithm is based on removing $L'$-irrelevant links from a set $L'$ until the size of $L'$ is bounded by $\OO(k^2 \lambda^2)$, and then branching on which of the remaining links is included in a solution.
We give two subroutines for finding irrelevant links.
Both of them need the following lemma for concluding that an instance is a no-instance under certain conditions.

\begin{lemma}
\label{lem:disjointseps}
Let $\ins = (G, L, \lambda, k)$ be an instance of the vertex connectivity augmentation problem, and $(A_1, B_1), \ldots, (A_\ell, B_\ell)$ a collection of $\ell > 2k$ proper separations of $G$ of order $< \lambda$, so that for every $(A_i, B_i)$ there exists a vertex $b_i \in B_i \setminus A_i$ that is not in any set $B_j$ for $j \neq i$.
Then, $\ins$ is a no-instance.
\end{lemma}

\begin{proof}
Assume that $\ell = 2k+1$.
For each $i \in [\ell]$, let $(A_i, B_i)$ be a separation of $G$ of order $< \lambda$ and $b_i$ a vertex so that $b_i \in B_i \setminus A_i$, and $b_i \in A_j \setminus B_j$ for all $j \neq i$.
Furthermore, assume that each $(A_i, B_i)$ minimizes $|B_i|$ (after fixing the vertices $b_1, \ldots, b_\ell$) among all such separations.

\begin{claim}
The sets $B_i \setminus A_i$ and $B_j \setminus A_j$ are disjoint for $i \neq j$.
\end{claim}
\begin{claimproof}
Suppose that $B_i \setminus A_i$ and $B_j \setminus A_j$ overlap.
Consider the separations $(A_i', B_i') = (A_i \cup B_j, B_i \cap A_j)$ and $(A_j', B_j') = (A_j \cup B_i, B_j \cap A_i)$.
Note that $b_i \in B_i' \setminus A_i'$ because $b_i \notin B_j$, and similarly $b_j \in B_j' \setminus A_j'$.

By submodularity (\Cref{lem:vertsubmod}),
\[|A_i' \cap B_i'| + |A_j' \cap B_j'| \le |A_i \cap B_i| + |A_j \cap B_j| < 2\lambda.\]

It follows that either the order of $(A_i', B_i')$ is $<\lambda$, or the order of $(A_j', B_j')$ is $< \lambda$.
In the former case, we could replace $(A_i, B_i)$ by $(A_i', B_i')$ while satisfying the requirement and making $|B_i|$ smaller, and in the latter case we could replace $(A_j, B_j)$ by $(A_j', B_j')$ while satisfying the requirement and making $|B_j|$ smaller.
So either of the cases would contradict the choice of $(A_1, B_1), \ldots, (A_\ell, B_\ell)$.
\end{claimproof}

Because the separations $(A_i, B_i)$ are proper and have order $< \lambda$, each of them must be crossed by at least one link in any solution to $\ins$.
However, as the sets $B_i \setminus A_i$ are disjoint, one link can cross at most two of the separations, so because $\ell > 2k$, there is no solution to $\ins$.
\end{proof}

We also need the following easy lemma about the existence of large independent sets on graphs of small average degree.

\begin{lemma}
\label{lem:avgdegindset}
Let $c \ge 1$ be an integer and $H$ a $q$-vertex graph with at most $c q$ edges.
Then $H$ contains an independent set of size $\ge q/(8c)$.
\end{lemma}
\begin{proof}
The graph $H$ has at most $q/2$ vertices of degree $\ge 4c$, because otherwise it would have more than $c q$ edges.
Therefore, $H$ has at least $q/2$ vertices of degree $< 4c$, so the greedy method of repeatedly selecting an arbitrary such vertex and discarding its neighbors results in an independent set of size $\ge \frac{q/2}{4c} \ge q/(8c)$.
\end{proof}

The first of our two subroutines for finding irrelevant links finds an irrelevant link among a large set of links all incident to one vertex, and the second one finds an irrelevant link among a large matching of links that satisfies certain conditions.
The first of the subroutines is inspired by Lemma~3.16 of~\cite{aug4vc}, while the second can be regarded as a more technical extension of the same ideas.

Now we give the subroutine for finding an irrelevant link among a set of links adjacent to one vertex.

\begin{lemma}
\label{lem:degirrelev}
Let $\ins = (G, L, \lambda, k)$ be an instance, $L' \subseteq L$ a set of links, and $v \in V(G)$ a vertex so that every link in $L'$ is incident to $v$.
If $|L'| \ge 20 \lambda k$, we can in time $(k+\lambda)^{\OO(1)} |G|$ either find an $L'$-irrelevant link or conclude that $\ins$ is a no-instance.
\end{lemma}
\begin{proof}
By taking a subset of $L'$, we can without loss of generality assume that $|L'| = 20 \lambda k = \ell$.
We denote the links in $L'$ by $v x_1, \ldots, v x_\ell$.
We start by introducing the sufficient condition for irrelevancy that we use.
We say that a link $v x_i \in L'$ is \emph{dominated} if there is no $(v, x_i)$-separation $(A,B)$ of order $<\lambda$ so that $x_j \in A$ for all $j \neq i$.

We claim that (1) a dominated link is $L'$-irrelevant, (2) if no link in $L'$ is dominated, then $\ins$ is a no-instance, and (3) we can detect if a link is dominated within the claimed running time.
The combination of these claims implies the lemma.

We start with (3).
\begin{claim}
\label{lem:degirrelev:c1}
We can detect if a link $v x_i$ is dominated in time $\OO(\lambda |G|)$.
\end{claim}
\begin{claimproof}
This can be done by using the Ford-Fulkerson algorithm to check if there is a separation $(A,B)$ of order $< \lambda$ so that $N[v] \cup \{x_j \mid j \neq i\} \subseteq A$ and $N[x_i] \subseteq B$.
\end{claimproof}

Then we prove (1).
\begin{claim}
\label{lem:degirrelev:c2}
If a link $v x_i$ is dominated, then it is $L'$-irrelevant.
\end{claim}
\begin{claimproof}
Suppose that $v x_i$ is dominated, and that $S \subseteq L$ is a solution to $\ins$ with $v x_i \in S$.
It suffices to show that there exists $j \neq i$ so that $S \cup \{v x_j\} \setminus \{v x_i\}$ is a solution to $\ins$.

Let $G' = G \cup S$ and $G'' = G' \setminus \{v x_i\}$.
If the vertex connectivity of $G''$ is $\ge \lambda$, then any $j \neq i$ can be chosen, so assume that the vertex connectivity of $G''$ is $< \lambda$.
Every proper separation $(A,B)$ of $G''$ of order $< \lambda$ must be an $(v, x_i)$-separation, because otherwise $(A,B)$ would also be a proper separation of $G'$, but $G'$ is $\lambda$-vertex-connected.
By \Cref{obs:vertconn} and the fact that $G'$ cannot be a clique because $|L'| > k$, the order of any such separation $(A,B)$ must be exactly $\lambda-1$.

By \Cref{lem:leftrightsep}, there exist $(v,x_i)$-separations $(A_1, B_1)$ and $(A_2, B_2)$ of $G''$ of order $\lambda-1$, so that for every $(v,x_i)$-separation $(A,B)$ of $G''$ of order $\lambda-1$, we have $A_1 \subseteq A \subseteq A_2$ and $B_2 \subseteq B \subseteq B_1$.

Because $v x_i$ is dominated and $(A_2,B_2)$ is also a separation of $G$, there exists $j \neq i$ so that $x_j \in B_2 \setminus A_2$.
This means that the link $v x_j$ crosses all $(v,x_i)$-separations of $G''$ of order $\lambda-1$, and in particular, it crosses all proper separations of $G''$ of order $<\lambda$.
Therefore, $G'' \cup \{v x_j\}$ is $\lambda$-vertex-connected, so $S \cup \{v x_j\} \setminus \{v x_i\}$ is a solution to $\ins$.
\end{claimproof}

We finish the proof by proving (2).
\begin{claim}
\label{lem:degirrelev:c3}
If no link in $L'$ is dominated, then $\ins$ is a no-instance.
\end{claim}
\begin{claimproof}
A link $v x_i$ is not dominated if and only if there exists a $(v,x_i)$-separation $(A,B)$ of order $< \lambda$ with $x_j \in A$ for all $j \neq i$.
For every $i \in [\ell]$, fix $(A_i, B_i)$ to be an arbitrary such separation.
Consider a graph on the vertex set $[\ell]$, having an edge between $i$ and $j$ if either $x_i \in B_j$ or $x_j \in B_i$.
Because $x_i \in B_j$ happens only if $x_i \in A_j \cap B_j$, and the $x_i$:s are disjoint, this graph has $< \lambda \cdot \ell$ edges.
By \Cref{lem:avgdegindset}, it has an independent set $X \subseteq [\ell]$ of size $|X| \ge \ell/(8\lambda) > 2k$.

For each $i \in X$, we have that $x_i \in B_i \setminus A_i$, and $x_i \notin B_j$ for $j \in X \setminus \{i\}$.
Now, \Cref{lem:disjointseps} implies that $\ins$ is a no-instance.
\end{claimproof}

The combination of \Cref{lem:degirrelev:c1,lem:degirrelev:c2,lem:degirrelev:c3} finishes the proof.
\end{proof}

For our second subroutine for finding irrelevant links, we need the following algorithmic subroutine.
We note to the reader that it may be motivating to read the proof of \Cref{lem:matchirrelev} before reading the proof of the following lemma in detail.
We also note that the algorithmic idea is similar to that of \Cref{lem:specialsepalgspecial}.

\begin{lemma}
\label{lem:matchirrelevsubroutine}
Let $G$ be a graph, $a,b \in V(G)$ two distinct vertices that are $(\lambda-k)$-vertex-connected for $\lambda,k \ge 0$, $T \subseteq V(G) \setminus \{a,b\}$ a set of terminal vertices, and $s$ an integer.
We can in time $(|T|+1)^k |G|^{\OO(1)}$ test if there exists an $(a,b)$-separation $(A,B)$ of order $< \lambda$ with $|T \cap (B \setminus A)| \le s$.
\end{lemma}
\begin{proof}
We give an algorithm that uses recursion by the parameter $k$, i.e., at each recursion level the value of $k$ decreases.

Let $\lambda'$ be the minimum order of an $(a,b)$-separation.
We can assume that $\lambda-k \le \lambda' < \lambda$.
By \Cref{lem:leftrightsep} there is an $(a,b)$-separation $(A_2, B_2)$ of order $\lambda'$ so that for all $(a,b)$-separations $(A,B)$ of order $\lambda'$ it holds that $A \subseteq A_2$ and $B_2 \subseteq B$.
We can find such a separation in $|G|^{\OO(1)}$ time.
Clearly, if a desired separation of order $\lambda'$ exists, then $(A_2, B_2)$ is such a separation.
Thus, we are done if $\lambda' = \lambda-1$.

Otherwise, we have that $\lambda' < \lambda-1$ and $|T \cap (B_2 \setminus A_2)| > s$.
Therefore, if $(A,B)$ is a desired separation, then for some $t \in T \cap (B_2 \setminus A_2)$ it must hold that $t \in A$.
We try all possibilities by recursion.
For fixed such $t$, let $G^t$ be the graph $G$ with the edge $at$ added.
We observe that the $(a,b)$-separations of $G^t$ correspond to the $(a,b)$-separations $(A,B)$ of $G$ with $t \in A$, so recursing to the graph $G^t$ for all possibilities $t \in T \cap (B_2 \setminus A_2)$ is correct.
To show that the running time of this recursive algorithm is $(|T|+1)^k |G|^{\OO(1)}$, it suffices to show that $a$ and $b$ are $(\lambda'+1)$-vertex-connected in $G^t$, implying that we can decrease $k$ when recursing.

To prove that $a$ and $b$ are $(\lambda'+1)$-vertex-connected in $G^t$, suppose otherwise, and let $(C,D)$ be an $(a,b)$-separation of order $\lambda'$ of $G^t$.
It is also an $(a,b)$-separation of $G$.
We have that $t \in C$ but $t \notin A_2$.
However, then $C \not\subseteq A_2$, which contradicts the fact that for all $(a,b)$-separations $(A,B)$ of order $\lambda'$ of $G$ it holds that $A \subseteq A_2$.
\end{proof}

Then we give a subroutine for finding an irrelevant link among a set of links that form a matching and cross a separation given by \Cref{lem:specialsep}.

\begin{lemma}
\label{lem:matchirrelev}
Let $\ins = (G, L, \lambda, k)$ be an instance, $(A,B)$ a separation of $G$ of order $< \lambda$ so that vertices in $A \setminus B$ are pairwise $\lambda$-vertex-connected, and $L' \subseteq L$ a set of links that (1) each cross $(A,B)$, and (2) form a matching.
If $|L'| \ge 40 \lambda k$, we can in time $2^{\OO(k \log (k+\lambda))} |G|^{\OO(1)}$ either find an $L'$-irrelevant link or conclude that $\ins$ is a no-instance.
\end{lemma}
\begin{proof}
We may assume that $G$ is $(\lambda-k)$-vertex-connected, as adding an edge can increase the vertex connectivity by at most $1$ if it is below $|V(G)|-2$.

By taking a subset of $L'$, we can assume that $|L'| = 40 \lambda k = \ell$.
We denote the links in $L'$ by $a_1 b_1, \ldots, a_\ell b_\ell$, with $a_i \in A \setminus B$ and $b_i \in B \setminus A$ for all $i$.
We say that a link $a_i b_i \in L'$ is \emph{dominated} if for every $(a_i, b_i)$-separation $(C,D)$ of order $< \lambda$, there exists at least $\lambda$ indices $j \in [\ell] \setminus \{i\}$ with $b_j \in D \setminus C$.

We claim that (1) a dominated link is $L'$-irrelevant, (2) if no link in $L'$ is dominated, then $\ins$ is a no-instance, and (3) we can detect if a link is dominated within the claimed running time.
The combination of these claims implies the lemma.

Let's start with (3).

\begin{claim}
\label{lem:matchirrelev:c1}
We can detect if a link $a_i b_i$ is dominated in time $2^{\OO(k \log (k +\lambda))} |G|^{\OO(1)}$.
\end{claim}
\begin{claimproof}
Note that $a_i b_i$ is \emph{not} dominated if and only if there exists an $(a_i, b_i)$-separation $(C,D)$ of order $< \lambda$ with less than $\lambda$ indices $j \in [\ell] \setminus \{i\}$ with $b_j \in D \setminus C$.
Finding out if such a separation exists corresponds to the problem of \Cref{lem:matchirrelevsubroutine}, which can be solved in $|L'|^k |G|^{\OO(1)} = 2^{\OO(k \log (k+\lambda))} |G|^{\OO(1)}$ time.
\end{claimproof}

Then we prove (1).

\begin{claim}
\label{lem:matchirrelev:c2}
If a link $a_i b_i$ is dominated, then it is $L'$-irrelevant.
\end{claim}
\begin{claimproof}
Suppose that $a_i b_i$ is dominated, and that $S \subseteq L$ is a solution to $\ins$ with $a_i b_i \in S$.
It suffices to show that there is $j \neq i$ so that $S \cup \{a_j b_j\} \setminus \{a_i b_i\}$ is a solution to $\ins$.

Let $G' = G \cup S$, and $G'' = G' \setminus \{a_i b_i\}$.
If the vertex connectivity of $G''$ is $\ge \lambda$, then any $j \neq i$ can be chosen, so assume that the vertex connectivity of $G''$ is $< \lambda$.
Every proper separation $(C,D)$ of $G''$ of order $< \lambda$ must be an $(a_i, b_i)$-separation, because otherwise $(C,D)$ would also be a separation of $G'$, but $G'$ is $\lambda$-vertex-connected.
By \Cref{obs:vertconn} and the fact that $G'$ cannot be a clique, the order of any such separation $(C,D)$ must be exactly $\lambda-1$.

By \Cref{lem:leftrightsep}, there exist $(a_i, b_i)$-separations $(C_1, D_1)$ and $(C_2, D_2)$ of $G''$ of order $\lambda - 1$ so that for every proper separation $(C, D)$ of $G''$ of order $\lambda-1$, we have $C_1 \subseteq C \subseteq C_2$ and $D_2 \subseteq D \subseteq D_1$.

Because $a_i b_i$ is dominated, for the separation $(C_2, D_2)$, there exists at least $\lambda$ indices $j \in [\ell] \setminus \{i\}$ with $b_j \in D_2 \setminus C_2$.
Because vertices in $A \setminus B$ are pairwise $\lambda$-vertex-connected, we have that $A \setminus B \subseteq C_1$, and in particular all but at most $\lambda-1$ vertices of $A \setminus B$ are in $C_1 \setminus D_1$.
Combining these two observations, we have that there exists $j \in [\ell] \setminus \{i\}$ with $a_j \in C_1 \setminus D_1$ and $b_j \in D_2 \setminus C_2$.
Therefore, $a_j b_j$ crosses all proper separations of $G''$ of order $< \lambda$, so $G'' \cup \{a_j b_j\}$ is $\lambda$-vertex-connected.
\end{claimproof}

Finally we prove (2).

\begin{claim}
\label{lem:matchirrelev:c3}
If no link in $L'$ is dominated, then $\ins$ is a no-instance.
\end{claim}
\begin{claimproof}
Recall that a link $a_i b_i \in L'$ is not dominated if and only if there exists an $(a_i, b_i)$-separation $(C,D)$ of order $< \lambda$ with less than $\lambda$ indices $j \in [\ell] \setminus \{i\}$ with $b_j \in D \setminus C$.
For each $i \in [\ell]$, fix $(C_i, D_i)$ be an arbitrary such separation.
Consider a graph on the vertex set $[\ell]$, having an edge between $i$ and $j$ if either $b_j \in D_i$ or $b_i \in D_j$.
Because there are $< \lambda$ indices $j$ with $b_j \in D_i \setminus C_i$, and $< \lambda$ indices $j$ with $b_j \in D_i \cap C_i$, this graph has $\le 2\lambda \cdot \ell$ edges.
By \Cref{lem:avgdegindset}, it has an independent set $X \subseteq [\ell]$ of size $|X| \ge \ell/(8 \cdot 2\lambda) > 2k$.

For each $i \in X$, we have that $b_i \in D_i \setminus C_i$, and $b_i \notin D_j$ for $j \in X \setminus \{i\}$.
Now, \Cref{lem:disjointseps} implies that $\ins$ is a no-instance.
\end{claimproof}

The combination of \Cref{lem:matchirrelev:c1,lem:matchirrelev:c2,lem:matchirrelev:c3} concludes the proof.
\end{proof}

\subsection{Putting things together}
We then combine \Cref{lem:specialsep,lem:degirrelev,lem:matchirrelev} into one lemma.

\begin{lemma}
\label{lem:togetherlemma}
Given an instance $\ins = (G,L,\lambda,k)$, we can in $2^{\OO(k \log(k+\lambda))} |G|^{\OO(1)}$ time obtain one of the following conclusions:
\begin{enumerate}
\item $G$ is $\lambda$-vertex-connected,
\item $\ins$ is a no-instance, or
\item return an $L$-relevant set $L' \subseteq L$ of size $|L'| \le \OO(\lambda^2 k^2)$.
\end{enumerate}
\end{lemma}
\begin{proof}
We first check for some easily testable special cases.
We test in $|G|^{\OO(1)}$ time if $G$ is $\lambda$-vertex-connected, and conclude with case~(1) if yes.
Otherwise, we first check if $\lambda \ge |V(G)|-2$, and in that case conclude with case~(3) returning $L' = L$.
Otherwise, we test in $|G|^{\OO(1)}$ time if $G$ is $(\lambda-k)$-vertex-connected, and if not, return that $\ins$ is a no-instance, i.e., case~(2).
This is justified by the fact that adding an edge can increase vertex connectivity by at most one when it is less than $|V(G)|-2$.


Then we use the algorithm of \Cref{lem:specialsepalg} to in time $2^k |G|^{\OO(1)}$ obtain a proper separation $(A,B)$ of $G$ of order $< \lambda$ so that vertices in $A \setminus B$ are pairwise $\lambda$-vertex-connected.
Now let $L_c \subseteq L$ be the links in $L$ that cross $(A,B)$.
Any solution must contain a link crossing $(A,B)$, so if $L_c$ is empty, we can return that $\ins$ is a no-instance, and otherwise $L_c$ is $L$-relevant.

Note that if $L_c'$ is a subset of $L_c$, and a link $uv \in L_c'$ is $L_c'$-irrelevant, then $L_c \setminus \{uv\}$ is $L$-relevant.
In this way we can repeatedly reduce the set $L_c$ by using \Cref{lem:degirrelev,lem:matchirrelev}, as long as we can either find a vertex with at least $20 \lambda k$ incident links in $L_c$, or a matching of size at least $40 \lambda k$ in $L_c$.
This keeps $L_c$ throughout $L$-relevant.

We first apply \Cref{lem:degirrelev} until all vertices have less than $20 \lambda k$ incident links in $L_c$.
After this, we note that if $|L_c| \ge 1600 \lambda^2 k^2$, the greedy method always finds a matching of at least $40 \lambda k$ links in $L_c$, so we can apply \Cref{lem:matchirrelev} to reduce $L_c$ until $|L_c| < 1600 \lambda^2 k^2$.
\end{proof}

Now, \Cref{thm:main} follows by branching.

\thmmain*
\begin{proof}
We design a recursive algorithm, where $k$ decreases at each recursion level.
Let $\ins = (G,L,\lambda,k)$ be the input instance.
If $k=0$, the problem is equivalent to checking if $G$ is $\lambda$-vertex-connected, which can be done in $|G|^{\OO(1)}$ time.

Otherwise, we apply \Cref{lem:togetherlemma} to either (1) conclude that $G$ is $\lambda$-vertex-connected, in which case we can output the empty set as the solution, (2) conclude that $\ins$ is a no-instance, in which case we return this fact, or (3) find an $L$-relevant set $L' \subseteq L$ of size $|L'| \le \OO(\lambda^2 k^2)$.

In the third case, if there exists a solution, then there exists a solution that contains a link from $L'$.
Therefore, for each $uv \in L'$, we recursively solve the instance $\ins^{uv} = (G \cup \{uv\}, L \setminus \{uv\}, \lambda, k-1)$, and if one of them returns a solution $S$, return the solution $S \cup \{uv\}$, and otherwise return that no solution exists.

The recursion tree has degree $\OO(\lambda^2 k^2)$ and depth $k$, so it has size $2^{\OO(k \log(k+\lambda))}$.
The running time at each node of the recursion tree is $2^{\OO(k \log(k+\lambda))} |G|^{\OO(1)}$, so the total running time is $2^{\OO(k \log(k+\lambda))} |G|^{\OO(1)}$.
\end{proof}

\section{Edge connectivity augmentation}
\label{sec:edgeconnaug}
In this section we prove \Cref{thm:edgeconn}.
The outline of the proof is the same as for vertex connectivity augmentation, except the algorithm is even simpler and we can avoid the dependency on $\lambda$ in the running time.

\subsection{Cut with a well-connected left side}
We start with a subroutine to find a cut of order $< \lambda$ with a certain structure.

Let $G$ be a graph and $\lambda$ an integer.
A $\lambda$-Gomory-Hu tree of $G$ is a triple $(T, \gamma, \alpha)$, where $T$ is a tree, $\gamma \colon V(G) \to V(T)$ maps vertices of $G$ to nodes of $T$, $\alpha \colon E(T) \to 2^{E(G)}$ maps edges of $T$ to sets of edges of $G$, and 

\begin{enumerate}
\item for all $u,v \in V(G)$ and $e \in E(T)$, if $\{e\}$ is a $(\gamma(u),\gamma(v))$-cutset of $T$, then $\alpha(e)$ is an $(u,v)$-cutset of $G$ of size $|\alpha(e)| < 
\lambda$, and
\item for all $u,v \in V(G)$, if $u$ and $v$ are not $\lambda$-edge-connected, then there exists $e \in E(T)$ so that $\{e\}$ is an $(\gamma(u), \gamma(v))$-cutset of $T$, and $\alpha(e)$ is a minimum-size $(u,v)$-cutset of $G$.
\end{enumerate}

The existence of $\lambda$-Gomory-Hu trees follows from Gomory-Hu trees~\cite{gomory1961multi}.
The $\lambda$-bounded version was previously defined e.g. by Hariharan, Kavitha, and Panigrahi~\cite{DBLP:conf/soda/HariharanKP07}.

\begin{theorem}[\cite{gomory1961multi}]
\label{thm:kgomoryhutree}
There is an algorithm that, given a graph $G$ and an integer $\lambda$, in time $|G|^{\OO(1)}$ computes a $\lambda$-Gomory-Hu tree of $G$.
\end{theorem}

We use $\lambda$-Gomory-Hu trees to find a cut with a certain structure.

\begin{lemma}
\label{lem:ensuresmalldeg}
Let $G$ be a graph and $\lambda$ an integer.
If $G$ is not $\lambda$-edge-connected, then $G$ has a proper cut $(A,B)$ of order $< \lambda$ so that vertices in $A$ are pairwise $\lambda$-edge-connected.
Furthermore, such a cut can be found in $|G|^{\OO(1)}$ time.
\end{lemma}
\begin{proof}
Let $(T,\gamma,\alpha)$ be a $\lambda$-Gomory-Hu tree of $G$.
We can assume that for every leaf $\ell$ of $T$, there exists a vertex $v \in V(G)$ so that $\gamma(v) = \ell$, as otherwise we could delete $\ell$ from $T$.
If $T$ has only one node, then $G$ is $\lambda$-edge-connected.
Otherwise, let $\ell$ be an arbitrary leaf, $A \subseteq V(G)$ the vertices $v$ with $\gamma(v) = \ell$, and $B = V(G) \setminus A$.
We have that $1 \le |A| < |V(G)|$, vertices in $A$ are pairwise $\lambda$-edge-connected, and $|E(A,B)| < \lambda$.
\end{proof}

\subsection{Irrelevant links}
Let $\ins = (G, L, \lambda, k)$ be an instance of the edge connectivity augmentation problem, and $L' \subseteq L$ a subset of the links.
A subset $L'' \subseteq L'$ is \emph{$L'$-relevant} if whenever there exists a solution to $\ins$ that contains a link from $L'$, there exists a solution that contains a link from $L''$.
We say that a link $uv \in L'$ is $L'$-irrelevant if $L' \setminus \{uv\}$ is $L'$-relevant.

For edge connectivity augmentation the situation is simpler than for vertex connectivity augmentation, as there is no need to require the links to vertex-disjoint as in the matching-case of vertex connectivity augmentation.
Therefore, we need only one irrelevancy lemma, which is next.


\begin{lemma}
\label{lem:edgeirrelev}
Let $\ins = (G,L,\lambda,k)$ be an instance, $L' \subseteq L$ a set of links, and $(A,B)$ a cut of $G$ of order $< \lambda$ so that all links in $L'$ cross $(A,B)$ and vertices in $A$ are pairwise $\lambda$-edge-connected.
If $|L'| > 2k$, we can in time $|G|^{\OO(1)}$ either find an $L'$-irrelevant link or conclude that $\ins$ is a no-instance.
\end{lemma}
\begin{proof}
We may assume that $|L'| = 2k+1 = \ell$.
Let us denote the links in $L'$ by $a_1 b_1, \ldots, a_\ell b_\ell$, so that $a_i \in A$ and $b_i \in B$ for all $i \in [\ell]$.
We say that a link $a_i b_i \in L'$ is \emph{dominated} if there is no $(a_i, b_i)$-cut $(C,D)$ of order $< \lambda$ so that $b_j \in C$ for all $j \in [\ell] \setminus \{i\}$.
Note that we can test if a link is dominated in $|G|^{\OO(1)}$ time by standard flow computation.

We claim that (1) a dominated link is $L'$-irrelevant and (2) if no link in $L'$ is dominated, then $\ins$ is a no-instance.
We start with (1).

\begin{claim}
\label{lem:edgeirrelev:c1}
If a link $a_i b_i$ is dominated, then it is $L'$-irrelevant.
\end{claim}
\begin{claimproof}
Suppose that $a_i b_i$ is dominated, and $S \subseteq L$ is a solution to $\ins$ with $a_i b_i \in S$.
It suffices to show that there is $j \neq i$ so that $S \cup \{a_j b_j\} \setminus \{a_i b_i\}$ is a solution to $\ins$.

Let $G' = G \cup S$ and $G'' = G' \setminus \{a_i b_i\}$.
If the edge connectivity of $G''$ is $\ge \lambda$, then any $j \neq i$ can be chosen, so assume that the edge connectivity of $G''$ is $< \lambda$.
It is exactly $\lambda-1$, because adding an edge can increase it by at most $1$.

Every proper cut $(C,D)$ of $G''$ of order $< \lambda$ must be a $(a_i,b_i)$-cut, because otherwise $(C,D)$ would also be a proper cut of $G'$ of order $< \lambda$.
By \Cref{lem:leftrightcut}, there exists $(a_i, b_i)$-cuts $(C_1, D_1)$ and $(C_2, D_2)$ of $G''$ of order $\lambda-1$, so that every $(a_i, b_i)$-cut $(C,D)$ of $G''$ of order $\lambda-1$ has $C_1 \subseteq C \subseteq C_2$ and $D_2 \subseteq D \subseteq D_1$.

Because $a_i b_i$ is dominated, there exists a link $a_j b_j$ with $j \in [\ell] \setminus \{i\}$ and $b_j \in D_2$.
Because $a_i$ and $a_j$ are $\lambda$-edge-connected, we have $a_j \in C_1$.
This means that $a_j b_j$ crosses all $(a_i, b_i)$-cuts of $G''$ of order $\lambda-1$, so it crosses all proper cuts of $G''$ of order $\lambda-1$, so $G'' \cup \{a_j b_j\}$ is $\lambda$-edge-connected.
Therefore, $S \cup \{a_j b_j\} \setminus \{a_i b_i\}$ is a solution to $\ins$.
\end{claimproof}

We then show (2).

\begin{claim}
\label{lem:edgeirrelev:c2}
If no link in $L'$ is dominated, then $\ins$ is a no-instance.
\end{claim}
\begin{claimproof}
Because no link in $L'$ is dominated, for each $a_i b_i \in L'$, there exists an $(a_i, b_i)$-cut $(C_i, D_i)$ of order $< \lambda$ so that $b_j \in C_i$ for all $j \neq i$.
Let $(C_1, D_1), \ldots, (C_\ell, D_\ell)$ be such cuts, selected so that each minimizes $|D_i|$.
We claim that the sets $D_i$ are pairwise disjoint.

If $D_i$ and $D_j$ overlap for $i \neq j$, then consider the cuts $(C_i', D_i') = (C_i \cup D_j, D_i \cap C_j)$ and $(C_j', D_j') = (C_j \cup D_i, D_j \cap C_i)$.
By submodularity (\Cref{lem:edgesubmod}),
\[|E(C_i', D_i')| + |E(C_j', D_j')| \le |E(C_i, D_i)|+|E(C_j, D_j)| < 2\lambda.\]

It follows that either $|E(C_i', D_i')| < \lambda$ or $|E(C_j', D_j')| < \lambda$.
In the former case, we could replace $(C_i, D_i)$ by $(C_i', D_i')$, and in the latter case replace $(C_j, D_j)$ by $(C_j', D_j')$, in both cases contradicting the selection by minimizing $|D_i|$ for all $i \in [\ell]$.
Therefore, the sets $D_i$ for $i \in [\ell]$ are pairwise disjoint.

Now, for each $(C_i, D_i)$, any solution must have a link that crosses $(C_i, D_i)$.
In particular, the solution must contain a link touching every $D_i$, and as they are disjoint and there are $\ell > 2k$ of them, any solution must have more than $k$ links, so $\ins$ is a no-instance.
\end{claimproof}

\Cref{lem:edgeirrelev:c1,lem:edgeirrelev:c2} conclude the proof.
\end{proof}

\subsection{Putting things together}

We then put the ingredients together to give our algorithm.

\thmedgeconn*
\begin{proof}
We give a recursive algorithm where $k$ decreases at each recursion level.
Let $\ins = (G, L, \lambda, k)$ be the input instance.
If $k=0$, the problem is equivalent to checking if $G$ is $\lambda$-edge-connected, which can be done in $|G|^{\OO(1)}$ time.

Otherwise, we apply \Cref{lem:ensuresmalldeg} to find in $|G|^{\OO(1)}$ time a proper cut $(A,B)$ of $G$ of order $< \lambda$ so that vertices in $A$ are pairwise $\lambda$-edge-connected.
Let $L' \subseteq L$ be the links that cross $(A,B)$.
We have that $L'$ is an $L$-relevant set.
By repeatedly applying \Cref{lem:edgeirrelev} to either find an $L'$-irrelevant link or conclude that $\ins$ is a no-instance, we can either conclude that $\ins$ is a no-instance, or reduce $L'$ until $|L'| \le 2k$, while keeping $L'$ an $L$-relevant set.
In the former case we can return immediately.

In the latter case, if there exists a solution, there exists a solution that contains a link from $L'$.
Therefore, for each $uv \in L'$, we recursively solve the instance $\ins^{uv} = (G \cup \{uv\}, L \setminus \{uv\}, \lambda, k-1)$, and if one of them returns a solution $S$, return the solution $S \cup \{uv\}$, and otherwise return that no solution exists.

The recursion tree has degree $2k$ and depth $k$, so it has size $2^{\OO(k \log k)}$.
The running time at each node is $|G|^{\OO(1)}$, so the total running time is $2^{\OO(k \log k)} |G|^{\OO(1)}$.
\end{proof}

\section{Conclusions}
\label{sec:conclusions}
We established the fixed-parameter tractability of the two most basic NP-hard variants of connectivity augmentation problems.
We now conclude with open problems.

An obvious open question is whether vertex connectivity augmentation is FPT parameterized by $k$ only.
We do not have a conjecture on which direction it should go, both directions seem equally plausible to us.
Another question is whether the algorithms of \Cref{thm:main,thm:edgeconn} extend to the setting of weighted links, i.e., where we have to pick at most $k$ links but also minimize their total weight.
Our algorithms seem to extend to this problem in the limited setting where the weights are integers bounded by a function of the parameters, simply by running the irrelevancy lemmas for each weight class separately.
However, it is not clear whether the more general setting where the weights can be large is FPT.



\bibliographystyle{alpha}
\bibliography{main}

\end{document}